\RequirePackage{amsthm}
\documentclass[pdflatex,sn-mathphys-num]{sn-jnl}% Math and Physical Sciences Numbered Reference Style 
\usepackage{graphicx}%
\graphicspath{{img/}}
\usepackage{multirow}%
\usepackage{amsmath,amssymb,amsfonts}%
\usepackage{amsthm}%
\usepackage{mathrsfs}%
\usepackage[title]{appendix}%
\usepackage{xcolor}%
\usepackage{textcomp}%
\usepackage{manyfoot}%
\usepackage{booktabs}%
\usepackage{algorithm}%
\usepackage{algorithmicx}%
\usepackage{algpseudocode}%
\usepackage{listings}%
\usepackage[normalem]{ulem}
\usepackage{subcaption}
\usepackage{mathtools}
\usepackage{dirtytalk}
\usepackage{soul}

\usepackage{anyfontsize}

\usepackage{tikz}
\usetikzlibrary{arrows}

\usepackage{todonotes}

\theoremstyle{thmstyleone}%
\newtheorem{theorem}{Theorem}%  meant for continuous numbers
\newtheorem{proposition}[theorem]{Proposition}% 

\theoremstyle{thmstyletwo}%

\theoremstyle{thmstylethree}%
\newtheorem{definition}{Definition}%

\theoremstyle{remark}
\newtheorem{remark}{Remark}

\newcommand{\PD}[1]     {\todo[inline, color=blue!20]{ {\sf PD: #1}}}
\newcommand{\DG}[1]     {\todo[inline, color=purple!20]{ {\sf DG: #1}}}

\newcommand{\AJL}[1]     {\todo[inline, color=orange!30]{ {\sf AJL: #1}}}

\newcommand{\hide}[1]{}

\newcommand*{\showArXiv}{}% 

\DeclarePairedDelimiter\abs{\lvert}{\rvert}

\begin{document}

\title[ClusterGraph]{ClusterGraph: a new tool for visualization and compression of multidimensional data}

\author[1]{\fnm{Paweł} \sur{Dłotko}}
\email{pdlotko@impan.pl}

\author*[1]{\fnm{Davide} \sur{Gurnari}}
\email{dgurnari@impan.pl}

\author[1,2]{\fnm{Mathis} \sur{Hallier}}
\email{mathis.hallier@etu.utc.fr}

\author[3]{\fnm{Anna} \sur{Jurek-Loughrey}}
\email{a.jurek@qub.ac.uk }

\affil[1]{\orgdiv{Dioscuri Centre in Topological Data Analysis}, \orgname{Mathematical Institute, Polish Academy of Sciences}, \orgaddress{\city{Warsaw}, \country{PL}}}

\affil[2]{\orgdiv{Génie Informatique}, \orgname{Université de Technologie de Compiègne}, \orgaddress{\city{Compiègne}, \country{FR}}}

\affil[3]{\orgdiv{School of Electronics, Electrical Engineering and Computer Science}, \orgname{Queens University of Belfast}, \orgaddress{\city{Belfast}, \country{UK}}}


\abstract{
    Understanding the global organization of complicated and high dimensional data is of primary interest for many branches of applied sciences.
    It is typically achieved by applying dimensionality reduction techniques mapping the considered data into lower dimensional space.
    This family of methods, while preserving local structures and features, often misses the global structure of the dataset.
    Clustering techniques are another class of methods operating on the data in the ambient space.
    They group together points that are similar according to a fixed similarity criteria, however unlike dimensionality reduction techniques, they do not provide information about the global organization of the data.

    Leveraging ideas from Topological Data Analysis, in this paper we provide an additional layer on the output of any clustering algorithm.
    Such data structure, \emph{ClusterGraph}, provides information about the global layout of clusters, obtained from the considered clustering algorithm.
    Appropriate measures are provided to assess the quality and usefulness of the obtained representation.
    Subsequently the ClusterGraph, possibly with an appropriate structure--preserving simplification, can be visualized and used in synergy with state of the art exploratory data analysis techniques.
}

% \keywords{keyword1, Keyword2, Keyword3, Keyword4}



\maketitle

\section{Introduction}\label{intro}

\ifdefined\showComments
    \AJL{I think this paper will benefit from more motivation provided in this section. Much more could be said why such a visualisation is important providing specific examples from different applications}
\fi

High-throughput experiments are becoming extremely common in applied sciences.
Now more than ever, large high-dimensional datasets are generated in almost every laboratory calling for an automated and reliable way to extract new knowledge from them.
Let us fix a dataset $X$, usually embedded in a high dimensional space.
Standard dimension reduction techniques, including PCA~\cite{pca}, t-SNE~\cite{tsne}, \textsc{UMAP}~\cite{umap} and \textsc{PHATE}~\cite{phate} aim to find a low dimensional embedding of $X$ so that points that are close in $X$, are also close in the embedding.
However, preservation of the global organization of $X$ in general, and information about distances of far away points in particular, is a challenge for this methods.

Clustering techniques~\cite{statLearning, SAXENA2017664} on the other hand, based on a fixed similarity measure, provide a partition of the input dataset $X$.
%This very general methodology has been successfully applied to a number of tasks~\cite{DBLP:books/crc/aggarwal2013}. 
However, clustering itself does not provide information about either intra-- or inter-- cluster organization of points, therefore is not used to asses the global structure of the data.
The aim of this work is two bridge this two approaches by enriching the output of a clustering algorithm with additional information on the data's global organization.
\hide{
    \begin{figure}[!htb]
        \centering
        % \begin{subfigure}[b]{0.22\linewidth }
        % \centering
        %     \input{nmeth_img/line_1}
        %     \subcaption{}
        % \end{subfigure}
        % \hfill
        % \begin{subfigure}[b]{0.22\linewidth }
        % \centering
        %     \input{nmeth_img/line_2}
        %     \subcaption{}        
        % \end{subfigure}
        % \hfill
        % \begin{subfigure}[b]{0.22\linewidth }
        % \centering
        %     \input{nmeth_img/triangle_1}
        %     \subcaption{}
        % \end{subfigure}
        % \hfill
        % \begin{subfigure}[b]{0.22\linewidth }
        % \centering
        %     \input{nmeth_img/triangle_2}
        %     \subcaption{}
        % \end{subfigure}
        \begin{subfigure}[b]{\linewidth }
            \centering
            \definecolor{tab_blue}{HTML}{1f77b4}
\definecolor{tab_orange}{HTML}{ff7f0e}
\definecolor{tab_green}{HTML}{2ca02c}
\definecolor{tab_red}{HTML}{d62728}
\definecolor{tab_purple}{HTML}{9467bd}
\definecolor{tab_brown}{HTML}{8c564b}
\definecolor{tab_pink}{HTML}{e377c2}
\definecolor{tab_gray}{HTML}{7f7f7f}
\definecolor{tab_olive}{HTML}{bcbd22}
\definecolor{tab_cyan}{HTML}{17becf}
\begin{tikzpicture}[line cap=round,line join=round,>=triangle 45,x=1cm,y=1cm]

\clip(-1,-0.6) rectangle (6,2.6);

% \draw [red, dashed] (1,2)-- (4,2);
% \draw [red, dashed] (1,2)-- (5,0.5);
% \draw [red, dashed] (1,2)-- (3,1);
% \draw [red, dashed] (1,2)-- (4,0);
% \draw [red, dashed] (1,2)-- (4.5,1);

% \draw [red, dashed] (0,1)-- (4,2);
% \draw [red, dashed] (0,1)-- (5,0.5);
% \draw [red, dashed] (0,1)-- (3,1);
% \draw [red, dashed] (0,1)-- (4,0);

% \draw [red, dashed] (1,0)-- (4,2);
% \draw [red, dashed] (1,0)-- (5,0.5);
% \draw [red, dashed] (1,0)-- (3,1);
% \draw [red, dashed] (1,0)-- (4,0);
% \draw [red, dashed] (1,0)-- (4.5,1);

% \draw [line width=2pt] (1,2)-- (0, 1);
% \draw [line width=2pt] (0,1)-- (1,0);
% \draw [line width=2pt] (4,2)-- (3,1);
% \draw [line width=2pt] (3,1)-- (4,0);
% \draw [line width=2pt] (4,0)-- (5,0.5);
% \draw [line width=2pt] (5,0.5)-- (4.5,1);
% \draw [line width=2pt] (1,0)-- (1,2);
% \draw [line width=2pt] (4,0)-- (4.5,1);
% \draw [line width=2pt] (4,0)-- (4,2);
% \draw [line width=2pt] (4,2)-- (5,0.5);
% \draw [line width=2pt] (4,2)-- (4.5,1);
% \draw [line width=2pt] (4.5,1)-- (3,1);
% \draw [line width=2pt] (3,1)-- (5,0.5);

\draw [rotate around={0:(1,2)},line width=2pt,color=tab_blue,fill=tab_blue,fill opacity=0.5] (1,2) ellipse (1.1180339887498945cm and 0.5cm);
\draw [rotate around={0:(1,0)},line width=2pt,color=tab_orange,fill=tab_orange,fill opacity=0.5] (1,0) ellipse (1.118033988749896cm and 0.5cm);
\draw [rotate around={0:(4,2)},line width=2pt,color=tab_purple,fill=tab_purple,fill opacity=0.5] (4,2) ellipse (1.118033988749901cm and 0.5cm);
\draw [rotate around={90:(3,1)},line width=2pt,color=tab_red,fill=tab_red,fill opacity=0.5] (3,1) ellipse (1.118033988749901cm and 0.5cm);
\draw [rotate around={0:(4,0)},line width=2pt,color=tab_cyan,fill=tab_cyan,fill opacity=0.5] (4,0) ellipse (1.1180339887498978cm and 0.5cm);
\draw [rotate around={90:(5,0.5)},line width=2pt,color=tab_pink,fill=tab_pink,fill opacity=0.5] (5,0.5) ellipse (0.7071067811865425cm and 0.5cm);
\draw [rotate around={0:(4.5,1)},line width=2pt,color=tab_gray,fill=tab_gray,fill opacity=0.5] (4.5,1) ellipse (0.6403124237432976cm and 0.4cm);
\draw [rotate around={90:(0,1)},line width=2pt,color=tab_green,fill=tab_green,fill opacity=0.5] (0,1) ellipse (1.118033988749896cm and 0.5cm);

\end{tikzpicture}
            \subcaption{}
        \end{subfigure}
        \hfill
        \vspace{0.04\linewidth}
        \hfill
        \begin{subfigure}[b]{0.3\linewidth }
            \centering
            \definecolor{tab_blue}{HTML}{1f77b4}
\definecolor{tab_orange}{HTML}{ff7f0e}
\definecolor{tab_green}{HTML}{2ca02c}
\definecolor{tab_red}{HTML}{d62728}
\definecolor{tab_purple}{HTML}{9467bd}
\definecolor{tab_brown}{HTML}{8c564b}
\definecolor{tab_pink}{HTML}{e377c2}
\definecolor{tab_gray}{HTML}{7f7f7f}
\definecolor{tab_olive}{HTML}{bcbd22}
\definecolor{tab_cyan}{HTML}{17becf}

\begin{tikzpicture}[line cap=round,line join=round,>=triangle 45,x=1cm,y=1cm]

\clip(1.25,-0.1) rectangle (6,2.6);

\draw [red, dashed] (2.5,2)-- (4,2);
\draw [red, dashed] (2.5,2)-- (5,0.5);
\draw [red, dashed] (2.5,2)-- (3,1);
\draw [red, dashed] (2.5,2)-- (4,0);
\draw [red, dashed] (2.5,2)-- (4.5,1);

\draw [red, dashed] (1.5,1)-- (4,2);
\draw [red, dashed] (1.5,1)-- (5,0.5);
\draw [red, dashed] (1.5,1)-- (3,1);
\draw [red, dashed] (1.5,1)-- (4,0);

\draw [red, dashed] (2.5,0)-- (4,2);
\draw [red, dashed] (2.5,0)-- (5,0.5);
\draw [red, dashed] (2.5,0)-- (3,1);
\draw [red, dashed] (2.5,0)-- (4,0);
\draw [red, dashed] (2.5,0)-- (4.5,1);

\draw [line width=2pt] (2.5,2)-- (1.5, 1);
\draw [line width=2pt] (1.5,1)-- (2.5,0);
\draw [line width=2pt] (4,2)-- (3,1);
\draw [line width=2pt] (3,1)-- (4,0);
\draw [line width=2pt] (4,0)-- (5,0.5);
\draw [line width=2pt] (5,0.5)-- (4.5,1);
\draw [line width=2pt] (2.5,0)-- (2.5,2);
\draw [line width=2pt] (4,0)-- (4.5,1);
\draw [line width=2pt] (4,0)-- (4,2);
\draw [line width=2pt] (4,2)-- (5,0.5);
\draw [line width=2pt] (4,2)-- (4.5,1);
\draw [line width=2pt] (4.5,1)-- (3,1);
\draw [line width=2pt] (3,1)-- (5,0.5);

\draw [fill=tab_blue] (2.5,2) circle (2.5pt);
\draw [fill=tab_green] (1.5,1) circle (2.5pt);
\draw [fill=tab_orange] (2.5,0) circle (2.5pt);
\draw [fill=tab_purple] (4,2) circle (2.5pt);
\draw [fill=tab_red] (3,1) circle (2.5pt);
\draw [fill=tab_cyan] (4,0) circle (2.5pt);
\draw [fill=tab_pink] (5,0.5) circle (2.5pt);
\draw [fill=tab_gray] (4.5,1) circle (2.5pt);

\end{tikzpicture}
            \subcaption{}
        \end{subfigure}
        \hfill
        \begin{subfigure}[b]{0.3\linewidth }
            \centering
            \definecolor{tab_blue}{HTML}{1f77b4}
\definecolor{tab_orange}{HTML}{ff7f0e}
\definecolor{tab_green}{HTML}{2ca02c}
\definecolor{tab_red}{HTML}{d62728}
\definecolor{tab_purple}{HTML}{9467bd}
\definecolor{tab_brown}{HTML}{8c564b}
\definecolor{tab_pink}{HTML}{e377c2}
\definecolor{tab_gray}{HTML}{7f7f7f}
\definecolor{tab_olive}{HTML}{bcbd22}
\definecolor{tab_cyan}{HTML}{17becf}

\begin{tikzpicture}[line cap=round,line join=round,>=triangle 45,x=1cm,y=1cm]

\clip(1.25,-0.1) rectangle (6,2.6);

\draw [line width=2pt] (2.5,2)-- (1.5,1);
\draw [line width=2pt] (1.5,1)-- (2.5,0);
\draw [line width=2pt] (4,2)-- (3,1);
\draw [line width=2pt] (3,1)-- (4,0);
\draw [line width=2pt] (4,0)-- (5,0.5);
\draw [line width=2pt] (5,0.5)-- (4.5,1);

\draw [red, dashed] (2.5,0)-- (2.5,2);
\draw [red, dashed] (3,1)-- (5,0.5);
\draw [red, dashed] (4.5,1)-- (3,1);
\draw [red, dashed] (4,2)-- (4,0);
\draw [red, dashed] (4,0)-- (4.5,1);
\draw [red, dashed] (4,2)-- (5,0.5);
\draw [red, dashed] (4,2)-- (4.5,1);

\draw [fill=tab_blue] (2.5,2) circle (2.5pt);
\draw [fill=tab_green] (1.5,1) circle (2.5pt);
\draw [fill=tab_orange] (2.5,0) circle (2.5pt);
\draw [fill=tab_purple] (4,2) circle (2.5pt);
\draw [fill=tab_red] (3,1) circle (2.5pt);
\draw [fill=tab_cyan] (4,0) circle (2.5pt);
\draw [fill=tab_pink] (5,0.5) circle (2.5pt);
\draw [fill=tab_gray] (4.5,1) circle (2.5pt);

\end{tikzpicture}
            \subcaption{}
        \end{subfigure}
        \hfill
        \begin{subfigure}[b]{0.3\linewidth }
            \centering
            \definecolor{tab_blue}{HTML}{1f77b4}
\definecolor{tab_orange}{HTML}{ff7f0e}
\definecolor{tab_green}{HTML}{2ca02c}
\definecolor{tab_red}{HTML}{d62728}
\definecolor{tab_purple}{HTML}{9467bd}
\definecolor{tab_brown}{HTML}{8c564b}
\definecolor{tab_pink}{HTML}{e377c2}
\definecolor{tab_gray}{HTML}{7f7f7f}
\definecolor{tab_olive}{HTML}{bcbd22}
\definecolor{tab_cyan}{HTML}{17becf}

\begin{tikzpicture}[line cap=round,line join=round,>=triangle 45,x=1cm,y=1cm]

\clip(1.25,-0.1) rectangle (6,2.6);

\draw [line width=2pt] (2.5,2)-- (1.5,1);
\draw [line width=2pt] (1.5,1)-- (2.5,0);
\draw [line width=2pt] (4,2)-- (3,1);
\draw [line width=2pt] (3,1)-- (4,0);
\draw [line width=2pt] (4,0)-- (5,0.5);
\draw [line width=2pt] (5,0.5)-- (4.5,1);

\draw [fill=tab_blue] (2.5,2) circle (2.5pt);
\draw [fill=tab_green] (1.5,1) circle (2.5pt);
\draw [fill=tab_orange] (2.5,0) circle (2.5pt);
\draw [fill=tab_purple] (4,2) circle (2.5pt);
\draw [fill=tab_red] (3,1) circle (2.5pt);
\draw [fill=tab_cyan] (4,0) circle (2.5pt);
\draw [fill=tab_pink] (5,0.5) circle (2.5pt);
\draw [fill=tab_gray] (4.5,1) circle (2.5pt);

\end{tikzpicture}
            \subcaption{}
        \end{subfigure}
        % \vspace{0.04\linewidth}

        \caption{\DG{shall we remove this picture?}The dataset in (a) is covered by eight overlapping clusters, whose union has two connected components.
            The ClusterGraph as a fully connected graphs is presented in (b).
            Note that it does not reflect the intrinsic shape of the set unlike the pruned ClusterGraph presented in (d).
            It is obtained from the ClusterGraph in (b) by isolating the cliques corresponding to different connected components (c) and then, within each of the components, removing edges that do not reflect the intrinsic distances between the data points.}
        \label{fig:example}
    \end{figure}
}
%%%%%%%

The first contribution of this paper is the construction of a \emph{ClusterGraph}; a graph--based structure on a top of a partition $\mathcal{C}(X)$ of the data obtained from a clustering algorithm $\mathcal{C}$ applied to $X$.
In the ClusterGraph $G = (V,E)$, each vertex corresponds to a single cluster from $\mathcal{C}(X)$.
Two vertices $u,v \in V$ are connected by an edge whose length corresponds to the distance between their respective clusters in $\mathcal{C}(X)$.
For the purpose of this construction, a number of inter--cluster distances defined in the ambient space are considered.

% \PD{If we decide to go with an example in Fig~\ref{fig:intuition} we can exchange the paragraph above with the following text:}
% \PD{Let us stop here to look at an example at Fig~\ref{fig:intuition} (a-d). In (a) and (c) $N$ points has been sampled from a two dimensional distributions centered in the blue points in panel (a) and (c). In both cases, we expect to obtain three clusters. What makes the two cases different is the global organization of the clusters; in the first case, the clusters form a triangular (Fig.~\ref{fig:intuition}(b)), in the second, a linear (Fig.~\ref{fig:intuition}(d)) shape. This information is encoded in ClusterGraphs when, for instance, an average distance between points in the clusters is used as a distance between clusters. ClusterGraphs build for the two point clouds presented above will be 3-cliques as long as three clusters are obtained. In the first (triangular) case, the length of all edges will be similar while in the second (linear) case, the lengths of two edges will be similar, while the third one will be twice as long. Hence the \emph{filtration} of the 3-clique with respect to the edge length, will be considerably different for those two cases. }

ClusterGraph has a number of advantages compared to alternative dimension reduction methods.
One of them is based on the fact that the distances, computed in the ambient space, are represented by labels on edges and not subjected to distortions made by standard dimension reduction techniques that force the projected data points to be embedded in an Euclidean space.
This  allows to visualize the \emph{global} distances in the dataset.
This is important as many datasets cannot be embedded into low dimensional Euclidean spaces without perturbing the distances between points.
% In addition, this method does not require for the input data points to be embedded in an Euclidean space, allowing to also tackle discrete metric sets. 
As an example of such a situation consider a collection of points in four clusters: $0$, $1$, $2$ and $3$. Points in each cluster are infinitesimally close. Distance between cluster $0$ and clusters $1$, $2$ and $3$ are $1$, while distance between  clusters $1$, $2$ and $3$ is $2$.
It is well known~\cite{morgan_embedding_1974, bourgain_lipschitz_1985}
%\DG{\url{https://mathoverflow.net/questions/12394/representability-of-finite-metric-spaces}} 
that such a graph cannot be isometrically embedded to any Euclidean space $\mathbb{R}^n$ for any $n$.
As a consequence, all dimension reduction techniques will distort the distances between clusters, as can be observed in Fig~\ref{fig:non_embedable_to_Euclidean}. ClusterGraph, on the contrary, provides the correct graph even in this case.

\begin{figure}
    \centering
    %  \begin{subfigure}[t]{0.45\textwidth }
    %      \centering
    %      \includegraphics[width=0.9\textwidth]{nmeth_img/impossible_CG.pdf}
    %      \caption{ClusterGraph}
    %      \label{fig:y_cg}
    %  \end{subfigure}
    %  %\vskip\baselineskip
    % \begin{subfigure}[t]{0.45\textwidth }
    %      \centering
    %      \includegraphics[width=0.9\textwidth, height=0.95\textwidth]{nmeth_img/impossible_umap.pdf}
    %      \caption{ \textsc{UMAP}} 
    %      \label{fig:y_cg}
    %  \end{subfigure}
    %  \begin{subfigure}[b]{0.45\textwidth }
    %      \centering
    %      \includegraphics[width=0.95\textwidth, height=0.9\textwidth]{nmeth_img/impossible_TSNE.pdf}
    %      \caption{t-SNE}
    %      \label{fig:y_cg}
    %  \end{subfigure}
    %  \begin{subfigure}[b]{0.45\textwidth }
    %      \centering
    %      \includegraphics[width=0.95\textwidth, height=0.9\textwidth]{nmeth_img/impossible_PHATE.pdf}
    %      \caption{PHATE}
    %      \label{fig:y_cg}
    %  \end{subfigure}
    % \includegraphics[width=0.7\textwidth]{nmeth_img/impossible.pdf}
    \includegraphics[width=\textwidth]{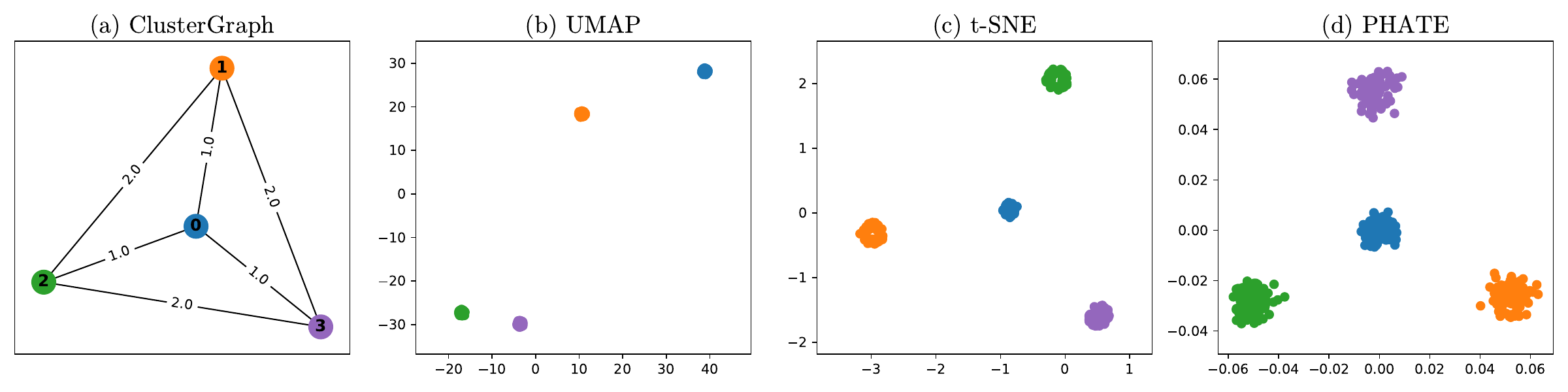}
    \caption{The dataset consisting of four clusters $0$ (blue), $1$ (orange), $2$ (green) and $3$ (purple), as described in the text, so that elements of cluster $0$ are distance one from elements from the remaining clusters and the mutual distances between elements of clusters $1, 2, 3$ are two.
        Such a dataset cannot be embedded, with the distances preserved, to any Euclidean space. In this case, UMAP (panel b) fails to capture the global layout, while t-SNE (panel c) and PHATE (panel d) do. However the coordinate systems of t-SNE and PHATE are drastically different.
        % \AJL{something wrong with this sentences}
        In both cases, as a result of the embedding into the Euclidean plane, the ratio of the distances $d(1,2) / d(0,1)$ is roughly $\sqrt{3}$ instead of the original $2$, the same is true for the other clusters. This is the optimal embedding that can be achieved when points are projected to Euclidean space. However, in the case of ClusterGraph (panel a), the distances are encoded as labels to the graph edges and therefore we are not restricted by any Euclidean coordinate system.
        % \AJL{it's not clear for me how a is better than c and d}  
        % \AJL{also, since this is the first figure you are referring to in text, perhaps it makes sense to make if figure 1.}
    }
    \label{fig:non_embedable_to_Euclidean}

\end{figure}

The second contribution is a method to assess the quality of the ClusterGraph $G$.
% obtained from point cloud $X$. 
Working under assumption that $X$ is sampled from a manifold equipped with an intrinsic distance, a \emph{metric distortion} between the intrinsic distances on $X$ and the distance induced by $G$ on $X$ is used to assess the quality of $G$.
The distance between points $x,y \in X$ induced by $G$ is the length of a shortest path in $G$ between vertices representing the clusters containing them.
% ($0$ if they belong to the same cluster)
% The intrinsic distance on $X$ is approximated using standard techniques~\cite{Crane2020ASO}. 
The logarithm of the ratio between the intrinsic and the ClusterGraph distance is used as a quality measure: the smaller its value, the better the quality of the ClusterGraph representation.

\ifdefined\showComments
    Note that the quality score defined above might be improved by removing certain edges from the considered ClusterGraph.
    \AJL{this may be discussed/explained later on, but at this point it is not clear why you would want remove edges from G. They represent distances between cluster. What would removing an edge mean in practice? No distance between clusters?}
    \DG{probably just remove the above sentence}
\fi
%%%%% HIDDEN
\hide{
    It can be observed in the example presented in Fig~\ref{fig:example} where a shape with intrinsic geometry is sampled (panel a).
    The full ClusterGraph (panel b) do not reveal an information about the structure of the set, however when edges with large metric distortion are removed (panel d), a shape of the initial set is recovered. \todo{remove if we decide to remove Fig1}
}
This procedure will be used to obtain a \emph{pruned ClusterGraph} that better approximates the intrinsic structure of the data.
For this purpose, a number of \emph{edge pruning} algorithms are proposed aiming to remove some edges while maintaining the global structure of the data.
A schema of the whole ClusterGraph pipeline is depicted in Fig.~\ref{fig:pipeline}.

% \begin{figure}[h]
%     \centering
%     \includegraphics[width=\textwidth]{nmeth_img/Pipeline_ClusterGraph.png}
%     \caption{Pipeline}
%     \label{fig:pipeline}
% \end{figure}

% \begin{figure}[h]
%     \centering
%     \input{nmeth_img/pipeline_CG}
%     \caption{ClusterGraph's Pipeline}
%     \label{fig:pipeline_test}
% \end{figure}

\begin{figure}[h]
    \centering
    \input{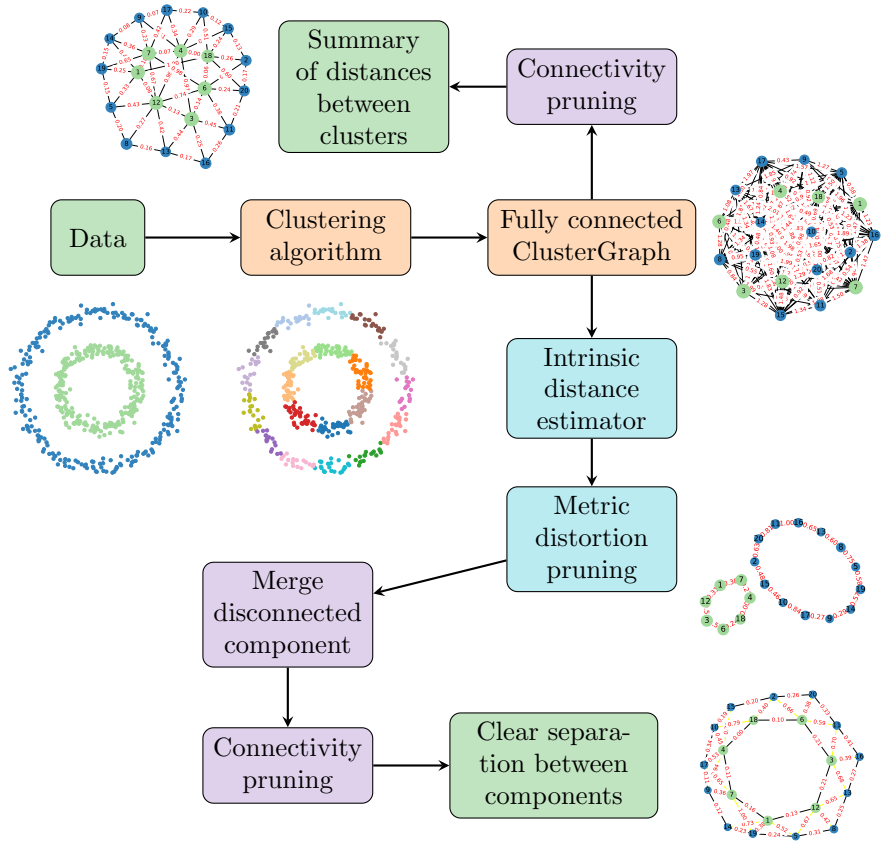}
    \caption{ClusterGraph pipeline with the two possible pruning strategies. Details on the example dataset that has been used to generate the figures can be found in Section~\ref{subsec:cc}.}
    \label{fig:pipeline}
\end{figure}

A Python implementation of the creation of ClusterGraph data structure, as well as the pruning algorithms and interactive visualization utilities is available at \href{https://github.com/dioscuri-tda/ClusterGraph}{github.com/dioscuri-tda/ClusterGraph}.

\section{Methods}\label{meth}
\subsection{ClusterGraph}~\label{sec:cg}
%In this section we present steps required to construct the \emph{ClusterGraph}  of a given dataset $X$ equipped with a metric $d_X : X \times X \rightarrow \mathbb R_{\ge 0} $\footnote{The conditions of the metric can be relaxed. The $d_X(x,y)=0$ iff $x=y$ as well as the triangle inequality are not necessary.}.  

Let $X$ be a dataset equipped with a metric $d_X : X \times X \rightarrow \mathbb R_{\ge 0}$.
%\footnote{The conditions of the metric can be relaxed. The $d_X(x,y)=0$ iff $x=y$ as well as the triangle inequality are not necessary.}.  
Take $\mathcal{C}$ to be an arbitrary hard or soft clustering algorithm.
Let $\mathcal{C}(X)$ be the partition or, in a more general case, a division of $X$ obtained from the clustering algorithm $\mathcal{C}$.
% \AJL{earlier you used C as partition, here you are referring to it as clustering algorithm. It is better to be consistent.}
A collection of sets $\{ C_i \}$ is a \emph{partition} of $X$ if for every  $C_i \neq C_j \in \mathcal{C}(X)$, $C_i \cap C_j = \emptyset$ and $\bigcup_{C_i \in \mathcal{C}(X)} C_{i} = X $.
Partitions of $X$ are typically obtained from hard clustering algorithms~\cite{statLearning}.
In a more general case, we can drop the empty intersection condition turning $\mathcal{C}(X)$ into a \emph{division} of $X$.
Divisions may be obtained using soft or fuzzy clustering methods~\cite{fuzzyClustering}.
They can also be obtained as byproducts of Topological Data Analysis techniques like Mapper~\cite{mapper} or Ball Mapper~\cite{ballmapper}.

In the next step we extend $d_X$ into a distance or similarity measure ${ d_\mathcal{C} : \mathcal{C}(X) \times \mathcal{C}(X) \rightarrow \mathbb R_{\ge 0} }$ defined on elements of $\mathcal{C}(X)$, as detailed in Section~\ref{sec:distances_clusters}.
The \emph{ClusterGraph} of a partition $\mathcal{C}(X)$ is a fully connected graph whose vertices are elements of $\mathcal{C}(X)$ and edges are weighted by distance $d_\mathcal{C}$ of the elements of $\mathcal{C}(X)$.

\ifdefined\showComments
    \begin{remark}
        \PD{I am storing the old text as a remark, putting the new below}
        The quality of the obtained graphs can be assessed by various criteria. In particular, in Section~\ref{sec:metric_distortion_quality_measure} we propose a \emph{metric distortion} based criteria to serve that purpose. ClusterGraph, as a fully connected graph, will often have a bad score that can be improved by removing certain edges from it.
        % as detailed in Section~\ref{sec:mdb_pruning}. 
        This way, a \emph{pruned} ClusterGraph will be obtained. We discuss different pruning strategies in Section~\ref{sec:pruning}.
    \end{remark}

    \PD{Above the old text, below the new text. Please check if it make sense}
\fi
%The ClusterGraph is mend to serve as a data visualization and compression method. Compression is achieved by collapsing points in the same elements of cluster to a vertex of the graph so that the nuber of vectices of the ClusterGraph is considerably smaller than the nuber of points. The intuition behind this step is that the points in the same cluster are located in a close proximity, and therefore ay be replaced with a single vertex of the graph. The visualization aspect is achieved when the layout of the ClustetGraph reasamble, in some sense, the laout of the input point cloud. That most likelly is not the case for a fully connected graph Therefore, in Section~\ref{sec:metric_distortion_quality_measure} we propose a \emph{metric distortion} based criteria and techniques allowing removial of certain edges from it so that the metric metric on a graph and on the inpital point cloud are becaming comparable. We discuss different strategies of removing edges from the ClusterGraph in Section~\ref{sec:pruning}. This way we obtain a \emph{pruned} ClusterGraph and use it for visualization purposes. 

ClusterGraph is intended to serve as a tool for data visualization and compression. Compression is achieved by collapsing points within the same cluster into a vertex of the ClusterGraph. This step is motivated by an assumption that points within the same cluster are in close proximity, hence share mulipe characteristics and can be prepresent by a single vertex. We expect ClusterGraph to have much fewer vertices compared to the original number of points. The visualization aspect is accomplished when the layout of the ClusterGraph resembles, to some extent, the layout of the input point cloud. However, this is unlikely to be the case for a fully connected graph. Therefore, in Section~\ref{sec:metric_distortion_quality_measure}, we propose a \emph{metric distortion}-based criterion and techniques for removing certain edges of the ClusterGraph so that the metrics on the graph and on the initial point cloud become comparable. Different strategies for edge removal from the ClusterGraph are discussed in Section~\ref{sec:pruning}. This process yields a \emph{pruned} ClusterGraph, which we use for visualization purposes.

\subsubsection{Distances between clusters}
\label{sec:distances_clusters}
In this section, for a dataset $X$ equipped with similarity measure $d_X$ and a partition $\mathcal{C}(X)$, a number of similarity measures $d_\mathcal{C} : \mathcal{C}(X) \times \mathcal{C}(X) \rightarrow \mathbb{R}_{\geq 0}$ are presented.
The choice of the optimal one is application dependent, very much like the clustering algorithm $\mathcal{C}$, and should therefore be selected and optimized by the user. Given two clusters $C_i$ and $C_j$, possible options include:

%Given two elements $C_i, C_j \in \mathcal{C}(X)$ we have implemented the following options for $d_\mathcal{C}(C_i,C_j)$;

\begin{enumerate}
    \item Maximum, minimum or an average distance between points
          \begin{align*}
              \min(C_i, C_j)                & = \min\limits_{x \in C_i , y \in C_j  } d_X(x,y)                                   \\
              \max(C_i, C_j)                & = \max\limits_{x \in C_i , y \in C_j  } d_X(x,y)                                   \\
              %\text{\DG{this makes no sense ,the distance between a cluster and itself it non zero!}} \\
              \operatorname{avg}(C_i , C_j) & = {\sum\limits_{x \in C_i} \sum\limits_{y \in C_j}  d_X(x, y)} / {(|C_i| |C_j|)} .
          \end{align*}

    \item Hausdorff distance
          \[
              d_{H}(C_i,C_j)=\max \left\{\,\sup _{x\in C_i}d(x,C_j),\,\sup _{y\in C_j}d(C_i,y)\,\right\}, \]
          where $d(a, B) = \inf\limits_{b \in B} d(a,b)$.
          %
          %\item Wasserstein distance,
          %
    \item Earth mover (a.k.a. Wasserstein) distance~\cite{emd, Wasserstein1969} which utilizes ideas from probability theory and optimal transport
          \[
              W_p(C_i, C_j) =  \inf_{\eta:C_i \rightarrow C_j} \left (\sum_{x \in C_i} d_X(x, \eta(x))^p \right)^\frac{1}{p},
          \]
          where $\eta$ is a matching between points of $C_i$ and $C_j$ and $1 \leq p < \infty$. If the two clusters have different size, the matching is computed in a weighted way, i.e. each point in $C_i$ is assigned a weight of $1/|C_i|$ such that each cluster has a total mass of 1.

\end{enumerate}

\subsubsection{Stability}

Let us consider a dataset $X$ , and two partitions of it $\mathcal{C}(X)$, $\mathcal{D}(X)$ obtained via some clustering algorithms. We are interest in quantifying how different these two partitions can be. In order to do so we introduce the following concept.

\begin{definition}[Image of a cluster]
    Let $X$ be a dataset and $\mathcal{C}(X)$, $\mathcal{D}(X)$ two partitions of it. The \emph{image} of a cluster $C_i \in \mathcal{C}(X)$ in $\mathcal{D}(X)$ is the union of all clusters of $\mathcal{D}(X)$ that contain some points of $C_i$, namely  $ \operatorname{im}_{\mathcal{D}(X)} (C_i) = \{ \bigcup D_j \in \mathcal{D}(X) \; | \; C_i \cap D_j \neq \emptyset \}$.
\end{definition}
This idea of mapping the points covered by one cluster in a given partition to the clusters in a second partition is inspired by an analogous technique for mapper graphs, \emph{MappingMappers}, described in~\cite{mappertype}.

Let us define the \emph{diameter} of a collection of points as the greatest distance between any pair of points. We can then state the following bound.

\begin{proposition}[Clustering stability]
    \label{prop:1}
    Let $X$ be a dataset and $\mathcal{C}(X)$, $\mathcal{D}(X)$ be two partitions of it such that the diameter of each set in  $\mathcal{C}(X)$ and $\mathcal{D}(X)$ is at most $\delta$. Then, for any cluster $C_i \in \mathcal{C}(X)$, its image in $\mathcal{D}(X)$ has diameter at most $3\delta$.
\end{proposition}
\begin{proof}
    Let $d_1, d_2 \in  X$ be the two points whose distance realizes the diameter of $\operatorname{im}(C_i)$. By definition of image there is at least one point $c_1 \in C_i$ which lies in the same cluster of $\mathcal{D}(X)$ as $d_1$, and similarly there exist at least one $c_2$ for $d_2$. Therefore we have
    \[
        d_X(d_1, d_2) \leq d_X(d_1, c_1) + d_X(c_1, c_2) + d_X(c_2, d_2) \leq 3\delta \; .
    \]
\end{proof}

Using a summary statistic of the distance between points as distance between clusters (option 1 in \ref{sec:distances_clusters}) allows us to derive a similar bound for ClusterGraphs built on top of $\mathcal{C}(X)$ and $\mathcal{D}(X)$.

\begin{definition}[Image of a ClusterGraph vertex]
    Let $X$ be a dataset, $\mathcal{C}(X)$, $\mathcal{D}(X)$ two partitions of it and $G_{\mathcal{C}(X)}$, $G_{\mathcal{D}(X)}$ the ClusterGraphs obtained from $\mathcal{C}(X)$ and $\mathcal{D}(X)$.
    The \emph{image} of a vertex $i \in G_{\mathcal{C}(X)}$ (corresponding to cluster $C_i \in G_{\mathcal{C}(X)}$) in $G_{\mathcal{D}(X)}$ is the collection of all vertices of $G_{\mathcal{D}(X)}$ that corresponds to clusters in $G_{\mathcal{D}(X)}$ containing some points of $C_i$.
\end{definition}

We define the \emph{diameter} of a weighted graph as the greatest distance between any pair of vertices.

\begin{proposition}[ClusterGraph stability]
    Let $X$ be a dataset, $\mathcal{C}(X)$, $\mathcal{D}(X)$ two partitions of it and $G_{\mathcal{C}(X)}$, $G_{\mathcal{D}(X)}$ the ClusterGraphs obtained from them.
    Assume that the diameter of each set in  $\mathcal{C}(X)$ and $\mathcal{D}(X)$ is at most $\delta$.
    Then the image of each vertex  $u \in G_{\mathcal{C}(X)}$ in $G_{\mathcal{D}(X)}$ is a clique of diameter at most $3\delta$ for the maximum and average distance, and $\delta$ for the minimum.
\end{proposition}
\begin{proof}
    Let $u$ be a vertex in $G_{\mathcal{C}(X)}$ and $\operatorname{im}(u)$ its image in $G_{\mathcal{D}(X)}$.
    % Let $v_i$ and $v_j$ be the vertices in $G_{\mathcal{D}(X)}$ corresponding to clusters $D_i$ and $D_j$ in $\mathcal{D}(X)$ whose distance is the diameter of $\operatorname{im}(u)$. 
    Recall that $\operatorname{im}(u)$ is a subset of the complete graph $G_{\mathcal{D}(X)}$, therefore it is a clique.
    Let $D_i$ and $D_j$ be the clusters in $\operatorname{im}(u)$ whose distance realizes the diameter of $\operatorname{im}(u)$, i.e. they correspond to the two vertices in the clique that are the furthest apart.

    Let us start with the maximum distance case.
    In particular, let $d_1 \in D_i$ and $d_2 \in D_j$ be the two data points realizing the maximum distance between $D_i$ and $D_j$, and therefore $d(d_1, d_2) = \operatorname{im}(u)$. Let $C_u$ be the cluster in $\mathcal{C}(X)$ corresponding to vertex $u \in G_{\mathcal{G}(X)}$. By definition of image of $u$ there are at least two points $c_1, c_2 \in C_u$ which lie in the same clusters of $\mathcal{D}(X)$ as $d_1$ and $d_2$, respectively.
    We can then proceed in a similar fashion to the proof of Proposition~\ref{prop:1}, namely we have
    \[
        \operatorname{diam}(\operatorname{im}(u)) = \max(D_i, D_j) = d_X(d_1, d_2) \leq d_X(d_1, c_1) + d_X(c_1, c_2) + d_X(c_2, d_2) \leq 3\delta \; .
    \]
    The same bound holds for the average distance since $\operatorname{avg}(D_i , D_j) \leq \max(D_i, D_j)$.

    For the minimum distance case it is sufficient to notice that $d_X(c_1, c_2) \leq \delta$ because they both belong to the same cluster $C_i$ whose diameter is bounded by $\delta$. Hence we have
    \[
        \operatorname{diam}(\operatorname{im}(u)) = \min(D_i, D_j) \leq d_X(c_1, c_2) \leq \delta .
    \]
\end{proof}

%%%%%%%%%%%%%%%%%%%%%%%%%%%%%%%%%%%%%%%%%%%%%%%%%%%%%%
\subsection{Metric distortion}
\label{sec:metric_distortion_quality_measure}

Given a dataset $X$, different choices of the clustering algorithm $\mathcal{C}$, as well as the metrics $d_X$ and $d_\mathcal{C}$ can lead to significantly different ClusterGraphs. The aim of this section is to introduce a score to assess the quality of a given
ClusterGraph $G$ by comparing it to the underlying geometric structure of the dataset $X$.

For this purpose let us assume that the considered point cloud $X$ is sampled from a compact and connected manifold $\mathcal{M}$ equipped with an  \emph{intrinsic distance} $d_{\mathcal{M}}$. Informally, the intrinsic distance between two points $x,y \in \mathcal{M}$ is  defined to be the infimum of the length of a curve $\gamma \subset \mathcal{M}$ joining $x$ and $y$, this is also known as \emph{geodesic} distance.
%Consider Fig~\ref{fig:intuition}(e) as a simple example. Taking the endpoints of the C-letter shape, their distance is the sum of lengths of the three line segments constituting the shape. 
In most applications, the underlying manifold is not known.
Consequently, the intrinsic distance needs to be estimated from the point cloud. This is a well-studied problem in computational geometry and computer graphics, and multiple methods have been proposed~\cite{isomap, klein_point_2004, ruggeri_approximating_2006,yu_geodesics_2014}.
% we have chosen to use shortest path distances on k-nearest neighbor graph of the point cloud $X$ as described in~\cite{isomap}. 
Below, we follow the approach of~\cite{isomap, Bernstein2000GraphAT} using the shortest path in the $k$-nearest neighbor graph as estimator.
Note that any other estimator of intrinsic distance can be also used in the proposed construction.

Let $G_{knn}(X)$ be the $k$-nearest neighbor graph on $X$ constructed in the following way: each point of $X$ corresponds to a vertex of $G_{knn}(X)$; it is connected to its $k$-nearest neighbors (in the chosen distance $d_X$, typically Euclidean), with $k$ being a parameter of the method. Weights corresponding to the distance between endpoints are assigned to the edges of $G_{knn}(X)$. We define a distance $d_X^k$ on $G_{knn}(X)$, estimating the intrinsic distance on $X$, as
\begin{equation}
    d_X^k(x,y) =  \text{the length of the shortest path between $x$ and $y$ in $G_{knn}(X)$ }
\end{equation}

\begin{remark}~\label{rmrk2}
    It may happen that $G_{knn}(X)$ is not connected.
    \ifdefined\showArXiv
        There are two possible reason for this, depicted in Fig.~\ref{fig:nn_example}.
    \else
        There are two possible reason for this.
    \fi
    In the first case  points are indeed sampled from a compact and connected manifold but the parameter $k$ is too low. This can be easily solved by increasing $k$. In the second case the underlying manifold is not connected. This will result in the $k$-nn graph being disconnected even for very high values of $k$, especially if many points are sampled. In this case, we will threat each connected component separately, splitting the input dataset $X$ (and the output of the clustering algorithm) into disjoint sets, each one corresponding to a different connected component and analyze each of them separately\footnote{By performing the construction in Sec.~\ref{sec:cg} we obtain a ClusterGraph for each connected component, each of them being a fully connected graph. In graph theory such a disjoint union of complete graphs is sometimes called a \say{cluster graph}. This unexpected but pleasing agreement in nomenclature motivates our choice of referring to our construction in camel case, to avoid confusion.}. For the rest of the Section we  therefore assume, without lack of generality, that the $k$-nn graph is fully connected. We discuss how to investigate the relations between different connected components in Section~\ref{sec:connectivity}.\\
\end{remark}

\begin{remark}~\label{rmrk1}
    Whenever an estimator is used, it is natural to ask how good such estimator is. The choice of a $k$-nn graph as an estimator of the geodesic distance is motivated by the following theorem by Bernstein,  Vin de Silva, Langford and Tenenbaum.

    \begin{theorem}[Theorem A in~\cite{Bernstein2000GraphAT}]

        Let $\mathcal{M}$ be a compact submanifold of $\mathbb{R}^n$, $X$ a finite set of data points in $\mathcal{M}$ and $G$ a graph on $X$ (for example, a $k$-nn graph). Then the inequalities
        \[
            (1-\lambda_1) d_\mathcal{M}(x, y) \leq d_G(x,y) \leq (1+\lambda_2) d_\mathcal{M}(x, y)
        \]
        are valid for all $x,y$ in $X$, where $\lambda_1, \lambda_2 < 1$ are two positive real numbers that depends on $G$, $\mathcal{M}$ and some technical assumptions on the density of $X$.

    \end{theorem}
\end{remark}

\ifdefined\showArXiv
    \begin{figure}[h]
        \centering
        \includegraphics[width=0.7\linewidth]{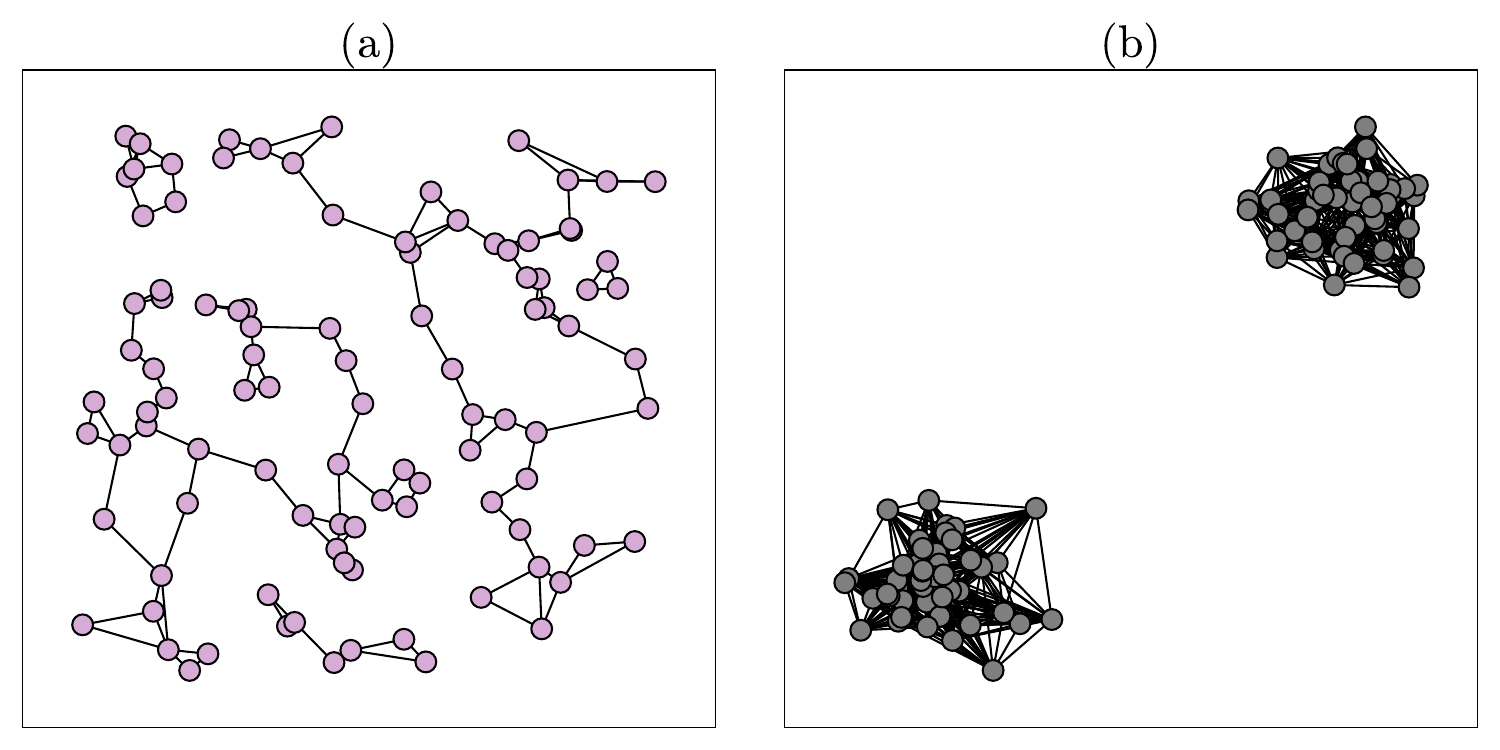}
        \caption{Two examples of a disconnected $k$-nn graph. Both panels contain $100$ points sampled from the unit square. In panel (a) the points are sampled from the uniform distribution and $k=2$. In panel (b) the points are sampled from two normal distributions centered in $(0.25, 0.25)$ and $(0.75, 0.75)$ with variance $0.1$ and $k=20$.}
        \label{fig:nn_example}
    \end{figure}
\fi

For each point $x \in X$, we denote with $C_x$ the cluster in $\mathcal{C}(X)$ that contains $x$. We can then use the distance between clusters described in \ref{sec:distances_clusters} to define a distance between points $d_{CG}$ in the ClusterGraph $G$ as following
\begin{equation}
    % d_{CG}(x,y) \coloneqq d_\mathcal{C}(C_x, C_y) .
    d_{CG}(x,y) = \text{the length of the shortest path between $C_x$ and $C_y$ in $G$. }
    \label{def:d_cg}
\end{equation}
Recall that $G$ is fully connected; one could wonder why the length of the shortest path between two vertices is used instead of the weight of the edge connecting them given by $d_\mathcal{C}$. First, the triangle inequality might not hold for $d_\mathcal{C}$. Secondly, we want the definition of $d_{CG}$ to hold also in the case of a pruned ClusterGraph, which we will discuss in Section~\ref{sec:pruning}.

This notion of $d_{CG}$ is well defined only when $\mathcal{C}(X)$ is a partition. In the more general case of a division, when a point can belong to more than one cluster, we take $d_{CG}(x,y)$ to be the length of the shortest path between any cluster containing $x$ and any cluster containing $y$.

Consider the ClusterGraph $G = (V, E)$ and let us fix two vertices $i,j \in V$, corresponding to two clusters $C_i$ and $C_j$.
For a pair of points $x \in C_i$ and  $y \in C_j$, $x \neq y$, we compute:
\begin{equation}
    \delta(x,y) = \abs{ \log \left(\frac{ d_{CG}(x,y) }{ d_X^k(x,y) }\right) } .
    \label{eq:distortion_pts}
\end{equation}
\ifdefined\showComments
    \AJL{why log here? Why are we penalizing more the cases where $d_{CG}$ is smaller than $d^k_X$. It can be either way, isn't it?}
    \PD{The use of the absolute value of the logarithm ensures that multiplicative scaling by a factor $\lambda$ between $d_{CG}(x,y)$ and $d_X^k(x,y)$ (in which case, $\frac{ d_{CG}(x,y) }{ d_X^k(x,y) } = \lambda$) as well as between $d_X^k(x,y)$ and $d_{CG}(x,y)$ (in which case, $\frac{ d_X^k(x,y) }{ d_{CG}(x,y) } = \lambda$) result in the same metric distortion $\log(\lambda)$.}
    \AJL{this may be a naive question. Why this need to be an intrinsic distance. I realize that for this formula it makes sense as they are both representing length of a path, but could it not be done some how if the dominator represent euclidean distance instead? Why does it make more sense to do it this way?}
\fi
The use of the absolute value of the logarithm ensures that a multiplicative scaling by a factor $\lambda = \tfrac{ d_{CG}(x,y) }{ d_X^k(x,y)} $ result in the same metric distortion as a $1/\lambda$ scaling, i.e. $\abs{\log(x/y)} = \abs{\log(y/x)}$ for every $x,y \in \mathbb{R}_{>0}$.

By averaging this quantity over all possible pairs of points $x \in C_i$ and $y \in C_j$, we obtain the metric distortion between clusters $C_i$ and $C_j$.

\begin{equation}
    \delta_{\{i,j\}}  = \frac{1}{|C_i| |C_j|} \sum\limits_{ (x,y) \in (C_i,C_j) }\delta(x,y)
    \label{eq:distortion_path}
\end{equation}
This score assesses how much the intrinsic distance between points on $X$ differs from their corresponding distance in the ClusterGraph.

The global metric distortion of the considered ClusterGraph can be obtained by averaging the score for each pair of vertices defined in Equation \ref{eq:distortion_path}. However, because clusters can have different sizes, we consider a weighted average. The weight of the pair $\{i,j\}$ corresponding to clusters $C_i$ and $C_j$ is defined as:
\begin{equation}
    w_{\{i,j\}} =  \frac{ |C_i \cup C_j|}{ (n-1) |X|  } ,
    \label{eq:weight_edge}
\end{equation}
where $n$ denotes the number of clusters in $\mathcal{C}(X)$ (or equivalently $n = |V|$, the number of vertices in the ClusterGraph). The rationale behind this averaging is to give more importance to interactions between large clusters.
% Further details are presented in the Appendix~\ref{app:proof_w}. 
The global metric distortion for a ClusterGraph (with respect to the $k$-nearest neighbors graph) can then be defined as
\begin{equation}
    \Delta_k(G)  = \frac{2}{n(n-1)} \sum\limits_{ \{i,j\} \in V } w_{\{i,j\}} \delta_{\{i,j\}} .
    \label{eq:global_metric_distortion}
\end{equation}
% \DG{this final paragraph needs to be rewritten}
% \hl{
% This quantity measures how well the ClusterGraph structure approximates a distance between points grouped in two clusters and thanks to it, we may informally say that it allows to detect "shortcuts" in the graph. Its integrated and weighted version presented in Equation}~\ref{eq:global_metric_distortion} \hl{gives a global score of the ClusterGraph. It indicates how well, in a metric sense, the given cluster graph represents the intrinsic metric structure of the data. As such, it can be treated as a p-value} \DG{really?? that's a strong claim} 
% \hl{
% attached to the ClusterGraph. In the next section we will use both the distortion of every single edge, and the global distortion of the graph, in the pruning process, aiming to remove certain edges so that the graph fits better the metric structure of the data (and therefore have a better score). 
% }

This quantity is a non-negative real number indicating how well the given ClusterGraph respects the intrinsic metric structure of the data.
\ifdefined\showComments
    \DG{do we want to add a comment about comparing the MD of two CG if they have equal number of clusters?}
\fi
It allows to compare the quality of ClusterGraphs having an equal, or very similar number of vertices. Comparison of metric distortions of ClusterGraphs having vastly different number of nodes should not be performed.

% Since it depends on the size of the dataset as well as the number of clusters, it should be used in a relative way when comparing different partitions of the same datasets with the same number of clusters. 
% \todo[inline]{is relative a good word? I want to say that it is not an absolute score, two distortions for two different datasets are not comparable}
In the next sections we will use each edge's distortion  as well as the global distortion to prune the ClusterGraph; with the goal of removing the edges that do not reflect the underlying structure of the data.
\\
\begin{remark}
    The ratio between two distances in our definition of the distortion (Eqn.~\ref{eq:distortion_pts}) might be reminiscent of the \emph{stretch factor} or \emph{distortion} of an embedding $f$. For two given points $x$ and $y$ in a metric space the stretch factor is defined as ${d(f(x) ,f(y))}/{d(x,y)}$.
    For example, consider a set of points in $\mathbb{R}^d$ and a connected graph having those points as vertex set. Each edge in the graph has a weight corresponding to the Euclidean distance between its endpoints.
    The stretch factor for two given points is the ratio of the length of the shortest path between them in the graph to their Euclidean distance. The stretch factor of the graph is the maximum stretch factor over any pair of points. Graphs with stretch factor at most $t$ are called \emph{$t$-spanners}~\cite{narasimhan_geometric_2007}.

    It is important to point out the differences between this widely studied topic in graph theory and our approach. First of all, we are not dealing with an embedding as the map that sends each data point to its cluster is highly non injective. Moreover, the stretch factor of a graph defined in the paragraph above is always greater than or equal to $1$. In our setting, the ratio between the ClusterGraph distance and the intrinsic one (Eqn.~\ref{eq:distortion_pts}) might be less than $1$, this is exactly the case of a \say{shortcut} edge in the ClusterGraph.
\end{remark}
\subsection{ClusterGraph pruning}
\label{sec:pruning}
ClusterGraph is, by definition, a fully connected graph. Consequently it may contain edges connecting regions of the datasets that are not close in the manifold $\mathcal{M}$ from which the data points of $X$ are sampled. We will refer to these edges as \say{shortcuts}, as they are shorter than the true geodesic distance between the corresponding points in the manifold and therefore are not representative of the underlying manifold structure.
The removal of such edges will make the ClusterGraph more similar, in the sense of metric distortion, to $X$. Moreover, the ClusterGraph may also contain edges the removal of which does not change considerably the metric structure of the graph. In this section we introduce three possible ideas to prune edges of a given ClusterGraph, and by doing so, of increasing the quality of the obtained representation.
% \begin{enumerate}
%     \item Greedy pruning, introduced in Section~\ref{sec:greedy_pruning},
%     \item Metric-distortion-based (MDB) pruning, introduced in Section~\ref{sec:mdb_pruning}
%     \item Distance-based (DB) pruning, introduced in Section~\ref{sec:distance_based_prunung}.
% \end{enumerate}
%
%
%
%
\subsubsection{Threshold pruning}
\label{sec:greedy_pruning}
If the triangle inequality holds for the distance between clusters $d_\mathcal{C}$, the length of the shortest path between each pair of vertices in the ClusterGraph (Eqn.~\ref{def:d_cg}) is exactly the length of the edge connecting them. It makes sense then to assign to each edge in the ClusterGraph the metric distortion for its two corresponding clusters, as defined in Equation~\ref{eq:distortion_path}.
% Intuitively, edges with large distortion can be considered \say{shortcuts} that are not representative of the underlying manifold structure from which the point cloud $X$ is sampled. 
The ClusterGraph can then be naively pruned by removing all edges having metric distortion greater than a threshold $\alpha > 0$.
\subsubsection{Iterative greedy pruning}
\label{sec:mdb_pruning}

Consider the ClusterGraph $G =(V,E)$. Let $\Delta_k(G)$ be its metric distortion as in Equation~\ref{eq:global_metric_distortion}. Denote with $G_{\hat{e}}$ the ClusterGraph obtained by removing edge $e$ from $G$, namely $G_{\hat{e}} = (V, E \setminus \{e\})$.
% Given that the distortion in Equations~\ref{eq:weight_edge} and \ref{eq:global_metric_distortion} are defined for a fully connected ClusterGraph, we need to generalize them to the non-complete case in the following way
% \begin{equation}
%     w_{\{i,j\}} =  \frac{ |C_i| + |C_j|}{ \sum_{v \in V} \deg(v) |C_v| } ,
% \label{eq:weight_edge_new}
% \end{equation}
% and
% \begin{equation}
%     \Delta_k(G)  = \frac{1}{|E|} \sum\limits_{ e \in E } w_e \delta_e .
%     \label{eq:global_metric_distortion_new}
% \end{equation}
% It is straightforward to see that one recovers Equations~\ref{eq:weight_edge} and \ref{eq:global_metric_distortion} in the special case of a fully connected graph with $n$ vertices. Moreover, the fact that the weights $w_{\{i,j\}}$ sum up to $1$ is a direct consecuence of the degree sum formula $\sum_{v \in V} \deg(v) = 2 |E|$.

We can then perform the following iterative greedy pruning procedure.
% compute $\Delta_k(G_{\hat{e}})$ for each $e \in E$. Then 
Remove edge $e$ if both conditions hold:
\begin{enumerate}
    \item $\Delta_k(G_{\hat{e}}) \leq \Delta_k(G)$,
    \item $\Delta_k(G_{\hat{e}}) \leq \Delta_k(G_{\hat{e}'})$ for any other $e' \in E$.
\end{enumerate}
Then update $E$ to $E \setminus \{e\}$.
The process may be repeated a fixed number of times, or as long as such an edge $e$ can be found. Note that the length of the shortest path between two vertices defined in Equation~\ref{def:d_cg} will be infinite if the ClusterGraph becomes disconnected, thus leading to an infinite value of the metric distortion. Therefore condition (1) ensures that the pruning procedure will never produce new connected components of the ClusterGraph.

% \begin{algorithmic}

% \For{$e_0 \in E$}
% {
%     \If{$\Delta_k(G_{\hat{e_0}}) \leq \Delta_k(G)$ } 
%         \If{$\Delta_k(G_{\hat{e_0}}) \leq \Delta_k(G_{\hat{e'}})$ for any other $e' \in E$}
%         \State $E \gets E \setminus \{e_0\}$
%         \State \textbf{break}
%     \EndIf
%     \EndIf 
% }
% \EndFor
% \end{algorithmic}

% In this approach, iterativelly, for every edge $e \in E$ we compute the metric discotrtion $D_e$ of $G$ with $e$ removed. An edge $e$ is removed from $G$ if;
% \begin{enumerate}
%     \item $D \geq D_e$,
%     \item $D_e \leq D_{e'}$ for any other $e' \in E$.
% \end{enumerate}
% The process may be repeated a fixed number of times, or as long as such an edge $e$ can be found.

\subsubsection{Connectivity based pruning}
\label{sec:connectivity}

The first two pruning techniques presented in Sections~\ref{sec:greedy_pruning} and~\ref{sec:mdb_pruning} focus on the removal of the edges with high metric distortion or, informally speaking, the removal all the \say{shortcuts} with respect to the structure of $X$. It might happen that after this pruning the obtained ClusterGraph still has a complicated structure which might render its visualization and interpretation challenging.

In what follows we adopt the \emph{connectivity based} approach by Zhou, Mahler and Toivonen~\cite{algorithms_prune} to the ClusterGraph pruning. A \emph{path} $P$ in $G=(V, E)$ is a set of edges $P = \{ \{i_1, i_2\}, \{ i_2, i_3 \}, \dots, \{ i_{k-1}, i_k \}  \} \in E$. A \emph{path quality function} $q(P) \rightarrow \mathbb{R}^+$ is defined by taking the sum of the inverse of the path length, calculated using the distance between clusters $d_\mathcal{C}$
\begin{equation}
    q(P) = \sum_{ \{i,j\} \in P} \frac{1}{d_\mathcal{C}(C_i, C_j)} .
\end{equation}

The \emph{connectivity} between two vertices $i,j$ in $G=(V,E)$ is the quality of the best path between them
\begin{equation}
    \operatorname{conn}(i,j;E) =
    \begin{cases}
        \max_{P \in \mathcal{P}(i, j)} q(P) & \text{if } \mathcal{P}(i, j) \neq \emptyset \\
        -\infty                             & \text{otherwise}
    \end{cases}
\end{equation}
where we denoted by $\mathcal{P}(i, j)$ the set of all possible paths between vertices $i$ and $j$.
The \emph{connectivity of a ClusterGraph} is the average connectivity over all pairs of vertices
\begin{equation}
    \operatorname{conn}(V,E) = \frac{2}{n(n-1)} \sum_{i,j \in V, i \neq j}\operatorname{conn}(i,j;E),
\end{equation}
where $n$ is the number of vertices in the ClusterGraph. Note that the connectivity will be $-\infty$ if $G$ is disconnected. In that case each connected component should be analyzed separately.

Let us now consider the ClusterGraph with one edge removed $G_{\hat{e}} = (V, E \setminus \{e\})$. It is straightforward to see that $\operatorname{conn}(V, E \setminus \{e\}) \leq \operatorname{conn}(V,E)$. In particular ${\operatorname{conn}(V, E \setminus \{e\}) = \operatorname{conn}(V,E)}$ if and only if $e$ does not belong to any of the best paths between any pairs of vertices. Moreover, $\operatorname{conn}(V, E \setminus \{e\}) =  -\infty$ if the removal of $e$ disconnects the graph. We can then define the \emph{ratio of connectivity kept} after removing an edge or, more generally, after removing a subset of edges $E_R \subset E$
\begin{equation}
    \mathit{rk}(V, E, E_R) = \frac{\operatorname{conn}(V, E \setminus E_R)}{\operatorname{conn}(V, E)}.
\end{equation}

Pruning can then be executed in a greedy iterative fashion (see Alg. 2 BF in~\cite{algorithms_prune}) by selecting, at each iteration, the edge whose the removal will result in smallest decrease of connectivity, i.e. the largest $\mathit{rk}$ value.

\subsection{Merging}
\label{sec:merging}
As discussed in Remark~\ref{rmrk2}, it might happen that the underlying manifold from which the data are sampled is disconnected.
In that case, the resulting ClusterGraph will have more than one connected component and each of them will be pruned separately.
In order to capture the global layout of the data, including the disconnected components of the manifold, we may merge different components by adding a collection of special edges between each vertex $v$ and its $k$-nearest neighbors not belonging to the same connected component as $v$.
Subsequently, the connectivity based pruning procedure can be applied to the newly added edges.

% \MH{Subsequently, the connectivity based pruning procedure can be applied to the edges added from the the merging process in order to get a clearer visualization of the relationship between those components. } \MH{after the merge, if the goal is to observe the intrinsic organization, we only prune the added edges
% }

% \subsection{Pipeline}
% \label{sec:pipeline}
% \DG{do we need such section?}

% \begin{algorithm}
% \caption{ClusterGraph pipeline}\label{alg:cap}
% \begin{algorithmic}
% \Require $\mathcal{C}$, intrinsic distance estimator
% \State Build the ClusterGraph $G$
% \For{each connected component $G_i$ of $G$}

%     \State Metric distortion pruning of $G_i$

% \EndFor
% \State Merge connected components
% \State Connectivity based pruning
% \end{algorithmic}
% \end{algorithm}

\section{Results}
\label{sec:results}

\ifdefined\showArXiv
    \subsection{Handwritten digits}

    We analyzed a copy of the test set of the UCI ML hand-written digits datasets~\cite{digits}, comprised of $1797$ grayscale images each one depicting a single digit. We partitioned the dataset into $100$ clusters using $k$-means. A ClusterGraph was built on top of the output of $k$-means using the average distance between points as distance between the corresponding clusters. The obtained ClusterGraph is therefore a fully connected graph with $4950$ edges. We used threshold pruning to simplify the graph, removing edges longer than a certain threshold $\alpha$. The resulting pruned ClustergGraphs for $\alpha=39$ and $\alpha=35$ are depicted in Fig.~\ref{fig:digits}.  $\alpha=39$ is the largest threshold for which the pruned graph is connected, while $\alpha=35$ produces a clear separation between clusters congaing different digits. Interactive visualization of such pruned ClusterGraphs allowing the user to inspect the images contained in each cluster are available at \href{https://github.com/dioscuri-tda/ClusterGraph}{github.com/dioscuri-tda/ClusterGraph}.

    \begin{figure}[h]
        \centering
        \begin{subfigure}[b]{0.4\textwidth }
            \centering
            \caption{}
            \vspace{-0.25cm}
            \includegraphics[width=\linewidth]{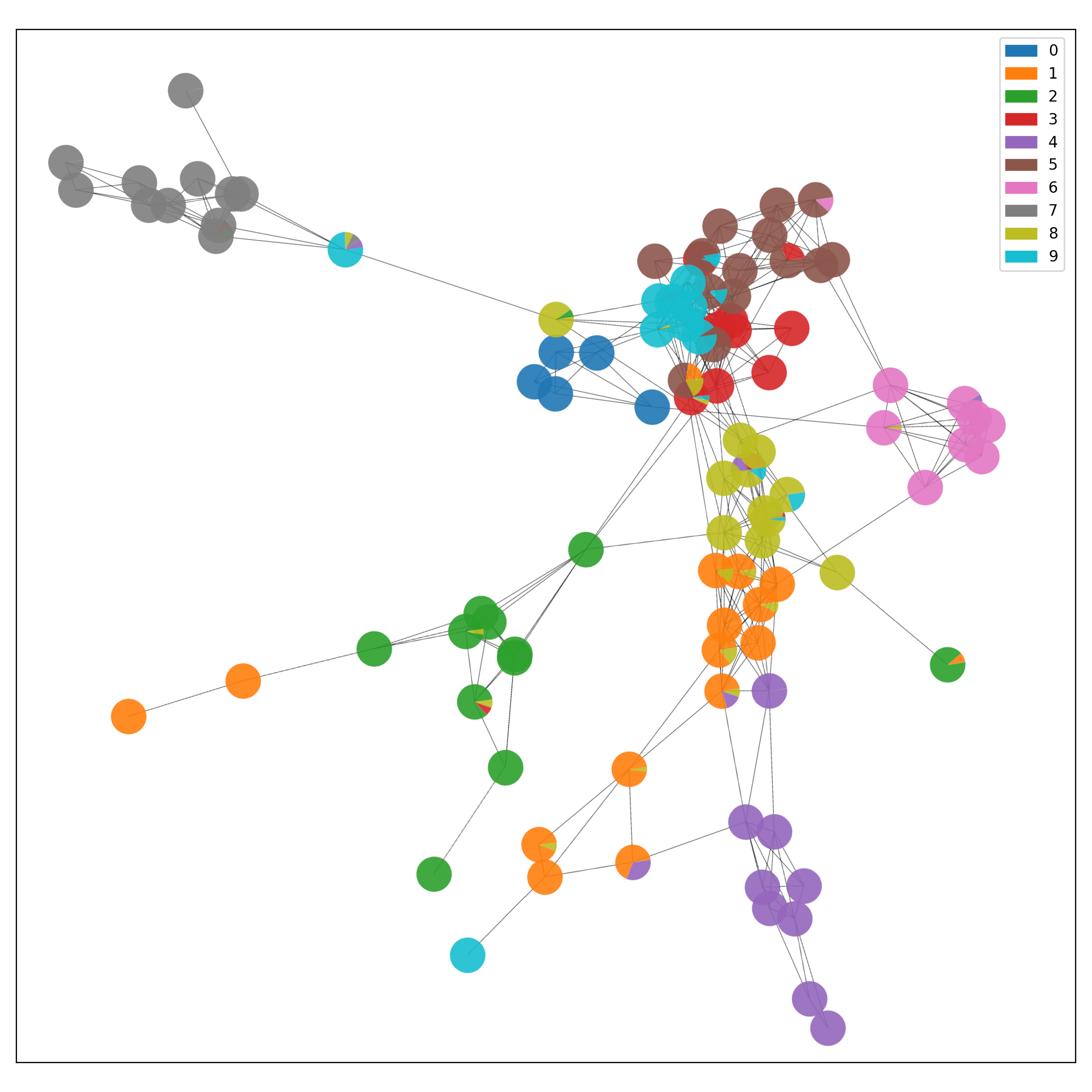}
        \end{subfigure}
        \hspace{0.25cm}
        \begin{subfigure}[b]{0.4\textwidth }
            \centering
            \caption{}
            \vspace{-0.25cm}
            \includegraphics[width=\linewidth]{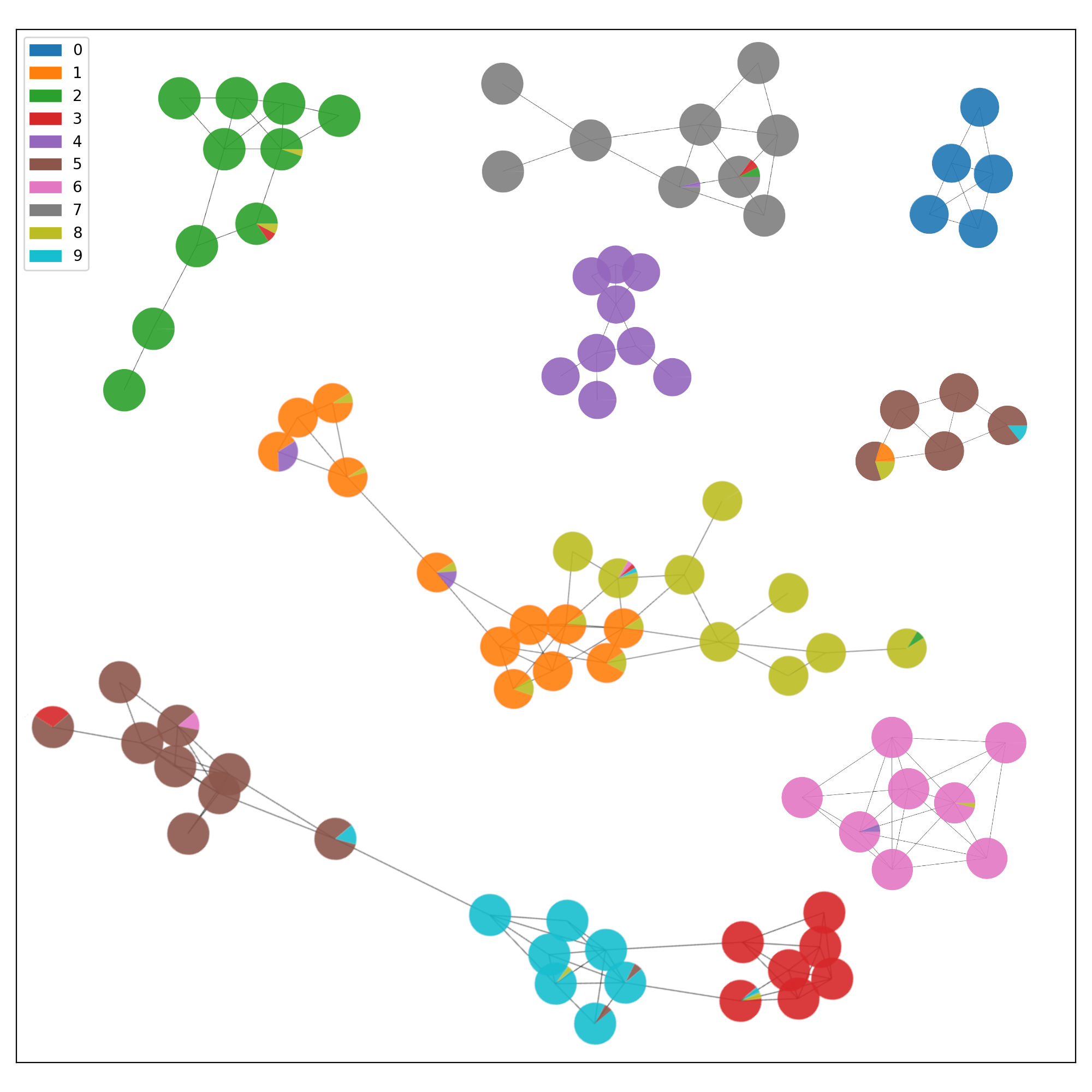}
        \end{subfigure}
        \caption{The two pruned ClusterGraphs obtained by removing all edges longer than $39$ (a) and $35$ (b). Each vertex in the ClusterGraph is depicted as a pie chart showing the percentage of points of each class in the corresponding cluster}
        \label{fig:digits}
    \end{figure}
\fi

\subsection{Concentric circles}
\label{subsec:cc}
To showcase the whole ClusterGraph pipeline (Fig.~\ref{fig:pipeline}) we considered $500$ points sampled from two concentric circles in the plane, depicted in Fig.~\ref{fig:circles}.
Clusters were computed using $k$-means with $20$ centroids and a ClusterGraph was built using the average Euclidean distance between points.
The $10$-nearest neighbors graph was used to estimate the intrinsic distance between data points.
The iterative metric distortion pruning procedure was applied and the pruned ClusterGraph is showed in Fig.~\ref{fig:circles}(b).
Note that the pruned ClusterGraph has two connected components, as a consequence of the $10$-nearest neighbors graph having two connected components.
Finally, the two components were merged by adding an edge between each vertex and its $3$-nearest neighbors in the other connected component.
We then pruned $20$ of these newly introduced edges using the connectivity based approach.
The resulting ClusterGraph is depicted in Fig.~\ref{fig:circles}(c).

\begin{figure}[h]
    \centering
    \includegraphics[width=\linewidth]{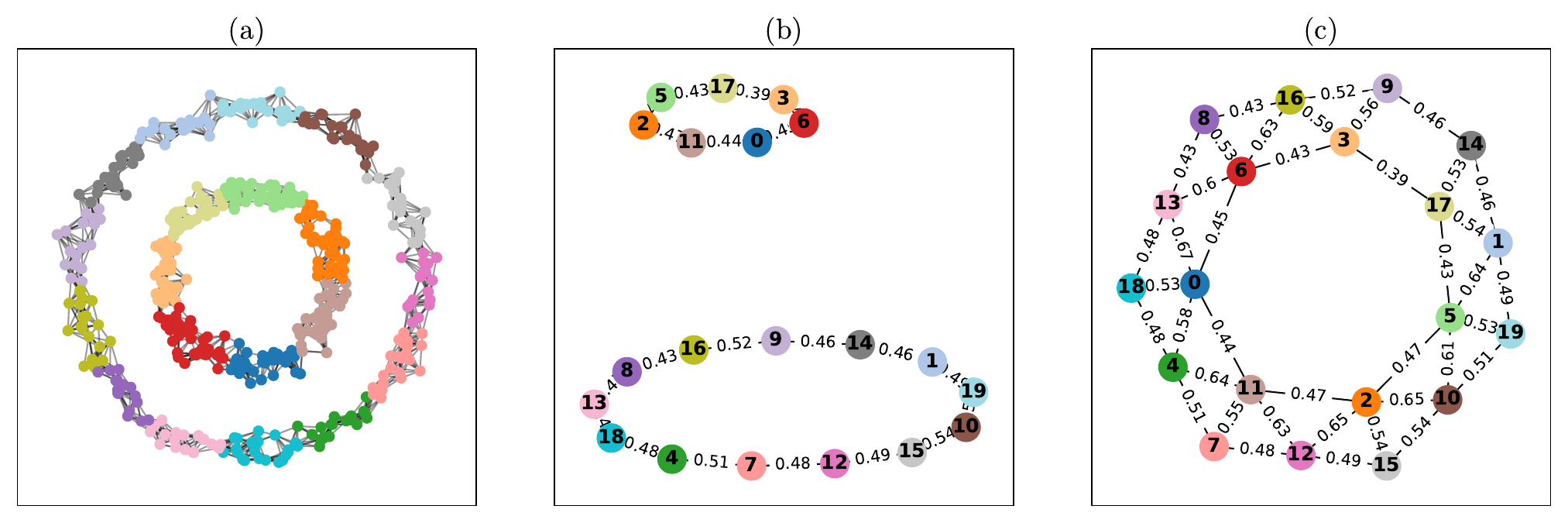}
    \caption{ClusterGraph built on the output of $k$-means for 500 points sampled from two concentric circles. The $20$ clusters are depicted in (a), on top of the $10$-nn graph. The two metric distortion-pruned components are depicted in (b) and subsequently merged and connectivity pruned (c). The vertices' colors are inherited from panel (a).
        \ifdefined\showComments
            \AJL{you are using k-nn graph as a baseline/reference point during pruning etc. Why not to use the k-nn graph then instead of cluster graph.}\PD{I think we have addressed this point due to the compression role of CG - let me know if it is good enough}
        \fi
    }
    \label{fig:circles}
\end{figure}

\subsection{Mice protein expression}

We analyzed the expression levels of $77$ proteins obtained from $38$ normal genotype control mice  and from $34$ of their trisomic littermates, both with and without treatment with the drug memantine and with and without the stimulation to learn~\cite{higuera_self-organizing_2015}. The original dataset contains $15$ measurements of each protein per sample, for a total of $1080$ data points.
% The eight groups learnt with different level of success. The outcome learning is divided into 4 classes. The mice which after the experiment managed to learn, the ones which did not learn, which failed and the ones with a rescued learning.
Control mice learn successfully while the trisomic ones fail, unless they are first treated with memantine, which rescues their learning ability.
%The data is divided in 4 classes : mice that were not stimulated to learn (\emph{no learning}, $555$), control mice that learned (\emph{normal}, $285$), not treated trisomic mice that failed to learn (\emph{failed}, $105$) and treated trisomic mice that learned successfully (\emph{rescued}, $135$).
The dataset is separated into four classes: mice that were not stimulated to learn (\emph{no learning}, $555$ samples), control mice that learned (\emph{normal}, $285$ samples), not treated and stimulated trisomic mice that failed to learn (\emph{failed}, $105$ samples) and treated and stimulated trisomic mice that learned successfully (\emph{rescued}, $135$ samples).
%
% The analysis was made on two differents groups of labels which are : the ability to learn (corresponding to the outcome of the experiment) and, the eight original groups of mice.

% \subsubsection{Comparison with dimensionality reduction techniques}

The Fig.~\ref{fig:mice} is showing the resulting ClusterGraph built on top of 18 clusters obtained via $k$-means~\cite{statLearning} alongside the output of popular dimensionality reduction techniques.
On Fig.~\ref{fig:mice}(a) which corresponds to ClusterGraph, one can quite clearly observe two main regions. One is almost only made of mice that were not stimulated to learn during the experiment. The second one is dominated by the normal learning and is also containing the rescued and failed learning. In such graph we can even observe a subregion dominated by the failed and rescued learning showing a difference between those proteins and the ones from the normal learning.
% \todo[inline]{I am not really sure about this paragraph, rewritings and suggestions are welcomed.}
% More than comparing the difference between proteins depending on the ability to learn, ClusterGraph let us observe the differences of the proteins depending on the type of treatment. One can observe that the shock to learn treatment has the strongest impact on proteins with one region made of proteins from mice which were stimulated (a group of close clusters in which more than 80\% of their data points were stimulated to learn) and another without stimulation ( a group of clusters with less than 20\% of points which were stimulated to learn). 
% Some differences can also be observed between trisomic and none trisomic mice with a group of clusters with less than 20\% of trisomic mice per cluster meanwhile another region is a lot more mixed. The region made of very few trisomic mice also corresponds to mice which were not shocked to learn. Such observation tends to show that the shock to learn minimizes the differences of proteins between trisomic and normal mice. When it comes to the memantine treatment, ClusterGraph is not showing any important difference between mice with memantine and without. 

\begin{figure}[ht!]
    \centering
    \includegraphics[width=\linewidth]{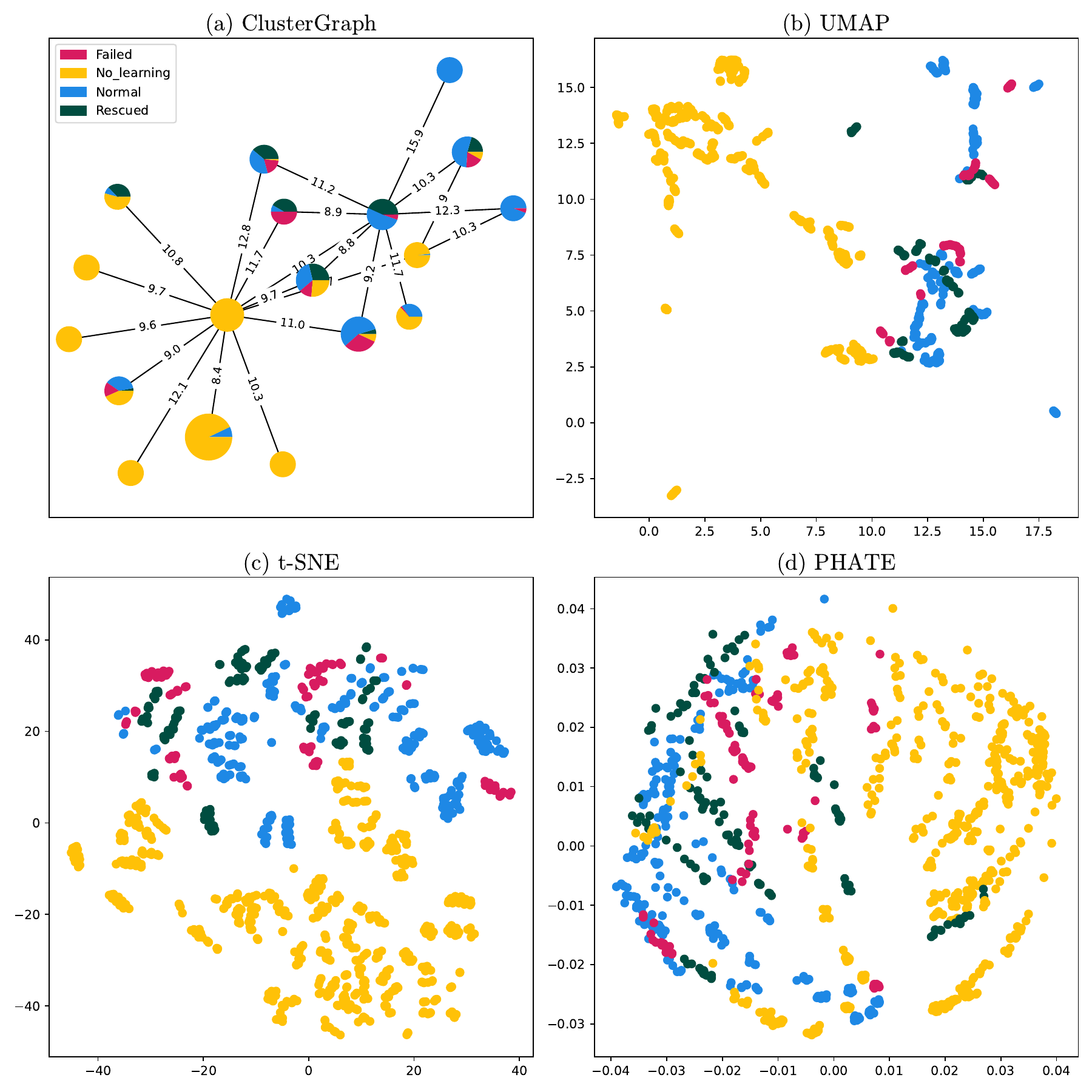}
    \caption{ClusterGraph visualization of the mice protein expression dataset alongside the output of other dimensionality reduction techniques. Each vertex in the ClusterGraph is depicted as a pie chart showing the percentage of points of each class in the corresponding cluster. The radius of each pie chart is proportional to the size of the corresponding cluster.}
    \label{fig:mice}
\end{figure}

\subsubsection{Assessing \textsc{UMAP}'s layout with ClusterGraph}

ClusterGraph can also be applied to the output of any dimensional reduction algorithm, in order to assess the quality of the returned low-dimensional embedding. As stated already, many of these methods aim to preserve the local structure of the point cloud, but they offer no guarantees on the global layout, as demonstrated in the following example.

We focus our attention on the \emph{failed} and \emph{rescued} classes. Variable selection using a random forest method~\cite{statLearning} was applied in order to identify the 10 most discriminating variables, this  10-dimensional pointcloud was then visualized using \textsc{UMAP} (Fig.~\ref{fig:mice_umap_clusters}(a)). Such projection is able to separate well the two classes, moreover some points appear to be outliers.

In order to quantify this observation, 10 clusters were selected from the two dimensional embedding using $k$-means (Fig.~\ref{fig:mice_umap_clusters}(b)). A ClusterGraph was then build on top of them using the distance between clusters in the original 10-dimensional space. The connectivity pruned ClusterGraph is depicted in Fig.~\ref{fig:mice_umap_clusters}(c), and it allows to compare the organization of the points in the original space versus the low dimensional \textsc{UMAP} embedding.
In both visualizations cluster 9 appears to be the central one, which is consistent with it being composed by a mixture of points from the two classes. We can however spot some clear differences between the two layouts.
Cluster 3, which is the outlier on the top left of the \textsc{UMAP} plot, is not an outlier in the ambient space, as it is closer to cluster 9 than, for example, cluster 6, which \textsc{UMAP} places close to the center.
Conversely, cluster 2, which appears to be near the centre in the \textsc{UMAP} plot, is at a significantly larger distance in the original space.

It is important to notice how both visualizations agree with respect to the \emph{local} layout, the differences appear a larger scales where \textsc{UMAP} fails to capture the global layout of the data.

\begin{figure}[h]
    \centering
    \includegraphics[width=\textwidth]{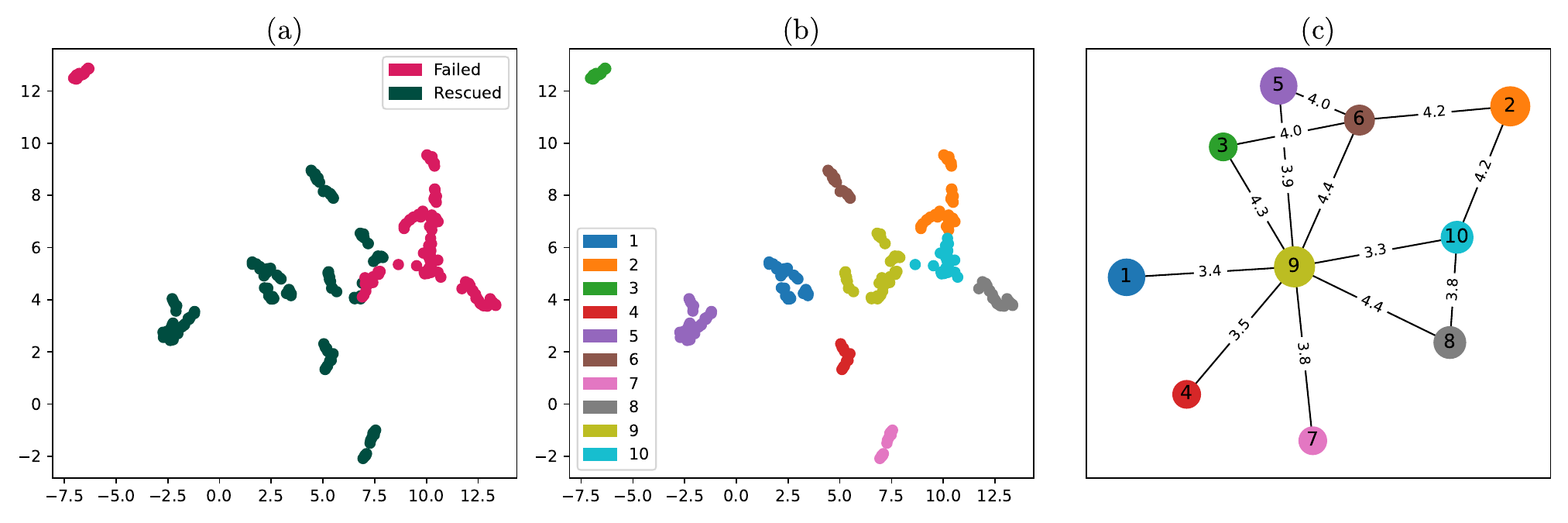}
    \caption{\textsc{UMAP} visualization of the \emph{failed} and \emph{rescued} samples is depicted in panels (a) and (b). Points are colored by class in the former and by the output of a $k$-means clustering in the latter. Panel (c) depicts the ClusterGraph obtained from such clustering, where the distances are computed on the original space.}
    \label{fig:mice_umap_clusters}
\end{figure}

\vspace{10cm}

\newpage
\subsubsection{Multi-level ClusterGraph}

In some scenarios we might be interested in further clustering a dataset which is already partitioned.
This is the case of our working example as each sample belongs to one of four classes: no learning, normal, failed or rescued.
A simple ClusterGraph obtained from such a coarse subdivision is depicted in Figure~\ref{fig:multi}(a).
Therefore we could cluster samples of each class separately thus obtaining a finer partition that still respect the class labels, i.e. all clusters are monochromatic with respect to the class label.

We partitioned the samples of each class into $5$ clusters by applying $k$-means to the data projected into the first $10$ principal components.
It is worth pointing out that, for each class, we computed the principal components by considering only the samples in that specific class.
This was done in order to obtain, for each class, a feature space that would emphasize the difference between samples in the same class.
Therefore the feature spaces of each the four classes that have been used in the $k$-means clustering are not not compatible with each others.

However, ClusterGraph can still be used to investigate the relative position of clusters belonging to different classes.
For this purpose we projected the entire dataset into its $10$ principal components, thus obtaining a global feature space.
We then built a ClusterGraph from the finer partition consisting of $5$ cluster for each one of the $4$ classes, using the average distance in the global feature space as distance between cluster.
A connectivity pruned version of this ClusterGraph is shown in Figure~\ref{fig:multi}(b).
Note how its global layout is similar to the ClusterGraph built on a partition of the whole dataset that does not respect the class labels, depicted in Figure~\ref{fig:mice}(a).
This multi-level approach allows us to obtain a clearer visualization by using the class labels as prior knowledge.

\begin{figure}[h]
    \centering
    % \begin{subfigure}{0.333\textwidth }
    %     \includegraphics[width=\textwidth]{nmeth_img/mice_CG_result_experiment.pdf}
    %     \label{fig:mice_pure}
    % \end{subfigure}
    % % \hfill
    % \begin{subfigure}{0.5\textwidth }
    %     \includegraphics[width=\textwidth]{nmeth_img/mice_CG_result_experiment_conn_pruned.pdf}
    % \label{fig:mice_pure_splitted}
    % \end{subfigure}

    \includegraphics[width=0.8\textwidth]{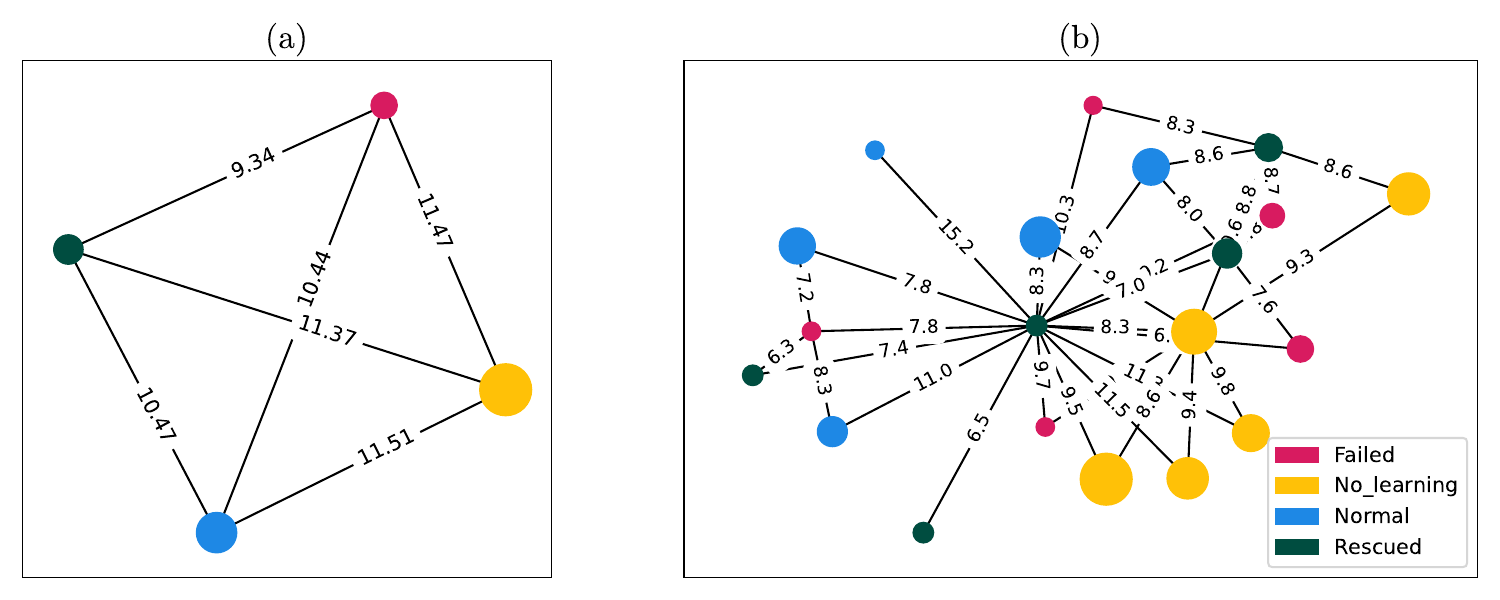}

    \caption{A simple ClusterGraph where each cluster corresponds to one of the four classes is depicted in panel (a). Each cluster can be further subdivided in 5 subclusters using $k$-means, the resulting connectivity pruned ClusterGraph is shown in panel (b). The radius of each vertex is proportional to the size of the corresponding cluster.}
    \label{fig:multi}

\end{figure}

\newpage
\section{Discussion}\label{discussion}
ClusterGraph is an effective tool for data visualization and compression. It can be constructed based on any clustering of the input data or arbitrary division of the data into labels. By introducing a notion of distance between clusters it provide a graph-based structure representing relationships within the data. Appropriate intrinsic-distance-preserving pruning of this graph results in a pruned graph that accurately represent the structure of the datasets. The quality of the obtained ClusterGraph can be assessed using the introduced concept of the metric distortion. It allows the analyst to assess the quality of ClusterGraph representation before further investigation.

We believe that ClusterGraph can serve as a comprehensive addition to the existing methods of dimensionality reduction and high dimensional data visualization.

\backmatter

\bmhead{Acknowledgements}
PD, DG and MH acknowledge support by the Dioscuri program initiated by the Max Planck Society, jointly managed with the National Science Centre (Poland), and mutually funded by the Polish Ministry of Science and Higher Education and the German Federal Ministry of Education and Research.

%===========================================================================================%%

\bibliography{references}% common bib file

\end{document}